\documentclass[USenglish,11pt]{article}
\usepackage[utf8]{inputenc}

\usepackage{amssymb,amsmath,amsthm}
\usepackage{url}
\usepackage{xspace}
\usepackage{enumitem}
\usepackage[dvipsnames,usenames]{xcolor}
\usepackage{tikz}
\usetikzlibrary{shapes}
\tikzset{
vertex/.style={circle,draw,inner sep=1pt,fill=Black}
}

\usepackage{multirow}

\title{The Descriptive Complexity of Subgraph Isomorphism without Numerics}

\author{Oleg Verbitsky%
 \thanks{Institut f\"ur Informatik,
  Humboldt-Universit\"at zu Berlin, Unter den Linden 6, D-10099 Berlin.
  Supported by DFG grant VE 652/1--2.
  On leave from the IAPMM, Lviv, Ukraine.}\and
Maksim Zhukovskii%
 \thanks{Moscow Institute of Physics and Technology, Laboratory of Advanced Combinatorics and Network Applications, Moscow.
 Supported by grants No.\ 15-01-03530 and 16-31-60052 of Russian Foundation for Basic Research.}}

\date{}

\theoremstyle{plain}
\newtheorem{theorem}{Theorem}[section]
\newtheorem{lemma}[theorem]{Lemma}
\newtheorem{proposition}[theorem]{Proposition}
\newtheorem{corollary}[theorem]{Corollary}
\theoremstyle{definition} 

\newcommand{\Case}[2]{\smallskip\par{\it Case #1:\/ #2}}

\newcounter{oq}
\newcommand{\que}{\refstepcounter{oq}\par{\sc \theoq.}~}

\newcommand{\refeq}[1]{(\ref{eq:#1})}

\newcommand{\feq}{\stackrel{\mbox{\tiny def}}{=}}

\newcommand{\classc}{\ensuremath{\mathcal C}\xspace}
\newcommand{\subgr}[1]{\ensuremath{\mathcal S(#1)}\xspace}

\newcommand{\tw}{\mathit{tw}}
\newcommand{\td}{\mathit{td}}
\newcommand{\Arb}{\mathrm{Arb}}

\newcommand{\prob}[1]{{\mathsf P[#1]}}
\newcommand{\rgraph}{\mathbb{G}_n}
\newcommand{\indsubgr}[1]{\ensuremath{\mathcal I(#1)}\xspace}

\newcommand{\probl}[1]{\textsc{\small #1}}
\newcommand{\logic}{{\mathcal L}}
\newcommand{\fo}{\ensuremath{\mathrm{FO}}\xspace}
\newcommand{\cclass}[1]{\textsf{\upshape #1}}
\newcommand{\ac}[1]{\cclass{AC$^{\cclass{#1}}$}\xspace}
\newcommand{\dobut}[2]{#1\cdot#2}
\begin{document}

\maketitle

\begin{abstract}
Let $F$ be a connected graph with $\ell$ vertices.
The existence of a subgraph isomorphic to $F$ can be defined
in first-order logic with quantifier depth no better than $\ell$,
simply because no first-order formula of smaller quantifier depth
can distinguish between the complete graphs $K_\ell$ and $K_{\ell-1}$.
We show that, for some $F$, the existence of an $F$ subgraph
in \emph{sufficiently large} connected graphs
is definable with quantifier depth $\ell-3$.
On the other hand, this is never possible with quantifier depth
better than $\ell/2$.
If we, however, consider definitions
over connected graphs with \emph{sufficiently large treewidth},
the quantifier depth can for some $F$ be arbitrarily small
comparing to $\ell$ but never smaller than the treewidth of~$F$.
Moreover, the definitions over \emph{highly connected graphs}
require quantifier depth strictly more than the density of $F$.
Finally, we determine the exact values of these descriptive complexity parameters
for all connected pattern graphs $F$ on 4 vertices.
\end{abstract}

\section{Introduction}\label{s:intro}

For a fixed graph $F$ on $\ell$ vertices,
let \subgr F denote the class of all graphs containing
a subgraph isomorphic to $F$.
The decision problem for \subgr F is known as
\probl{Subgraph Isomorphism} problem.
It is solvable in time $O(n^\ell)$ on $n$-vertex input graphs by exhaustive search.
Ne\v{s}et\v{r}il and Poljak \cite{NesetrilP85} showed that
\subgr F can be recognized in time $O(n^{(\omega/3)\ell +2})$, where $\omega<2.373$ is the exponent of fast square
matrix multiplication~\cite{Gall14}.
Moreover, the color-coding method by Alon, Yuster and
Zwick \cite{AlonYZ95} yields the time bound
$$
2^{O(\ell)}\cdot n^{\tw(F)+1}\log n,
$$
where $\tw(F)$ denotes the treewidth of $F$.
On the other hand, the decision problem for \subgr{K_\ell},
that is, the problem of deciding if an input graph contains a clique
of $\ell$ vertices, cannot be solved in time $n^{o(\ell)}$
unless the Exponential Time Hypothesis fails~\cite{ChenHKX06}.

We here are interested in the \emph{descriptive complexity} of \probl{Subgraph Isomorphism}.
A sentence $\Phi$ \emph{defines} a class of graphs $\classc$ if
\begin{equation}
  \label{eq:defines}
G\models\Phi\iff G\in\classc,
\end{equation}
where $G\models\Phi$ means that $\Phi$ is true on $G$.
For a logic $\logic$, we let $D_\logic(\classc)$ (resp.\ $W_\logic(\classc)$)
denote the minimum quantifier depth (resp.\ variable width) of $\Phi\in\logic$
defining $\classc$. Note that $W_\logic(\classc)\le D_\logic(\classc)$.
We simplify notation by writing
\begin{equation}
  \label{eq:WDF}
W_\logic(F)=W_\logic(\subgr F)\text{ and }D_\logic(F)=D_\logic(\subgr F).
\end{equation}

We are primarily interested in the first-order logic of graphs
with relation symbols for adjacency and equality of vertices,
that will be denoted by \fo. We suppose that the vertex set of any $n$-vertex graph
is $\{1,\ldots,n\}$. Seeking the adequate logical formalism for various models of computation,
descriptive complexity theory considers also more expressive logics involving
numerical relations over the integers. Given a set $\mathcal N$ of such relations,
$\fo[\mathcal N]$ is used to denote the extension of \fo whose language
contains symbols for each relation in $\mathcal N$. Of special interest are
$\fo[<]$, $\fo[+,\times]$, and $\fo[\Arb]$, where $\Arb$ indicates that
arbitrary relations are allowed. It is known \cite{Immerman-book,Libkin04}
that $\fo[\Arb]$ and $\fo[+,\times]$ capture (non-uniform) \ac0 and DLOGTIME-uniform \ac0
respectively.

We will simplify the notation \refeq{WDF} further by writing
$D(F)=D_\fo(F)$ and $W(F)=W_\fo(F)$. Dropping \fo in the subscript,
we also use notation like $D_<(F)$ or $W_\Arb(F)$. In this way
we obtain two hierarchies of width and depth parameters. In particular,
$$
W_\Arb(F)\le W_<(F)\le W(F)\text{ and }D_\Arb(F)\le D_<(F)\le D(F).
$$
The relation of $\fo[\Arb]$ to circuit complexity implies that
\subgr F is recognizable on $n$-vertex graphs by bounded-depth unbounded-fan-in circuits of
size $O(n^{W_\Arb(F)})$; see~\cite{Immerman-book,Rossman08}.
The interplay between the two areas has been studied in \cite{KawarabayashiR16,KouckyLPT06,LiRR14,Rossman08,Rossman16}.
Noteworthy, the parameters $W_\Arb(F)$ and $D_\Arb(F)$
admit combinatorial upper bounds
\begin{equation}
  \label{eq:W-tw}
W_\Arb(F)\le\tw(F)+3\text{ and }D_\Arb(F)\le\td(F)+2
\end{equation}
in terms of the treewidth and treedepth of $F$; see~\cite{Rossman-talk}.\footnote{In his presentation \cite{Rossman-talk},
Benjamin Rossman states upper bounds $W_\fo(F)\le\tw(F)+1$ and $D_\fo(F)\le\td(F)$ for the \emph{colorful} version of
\textsc{Subgraph Isomorphism} studied in \cite{LiRR14}.
It is not hard to observe that the auxiliary color predicates can be defined in $\fo[\Arb]$ at the cost
of two extra quantified variables by the color-coding method developed in \cite{AlonYZ95};
see also \cite[Thm.~4.2]{Amano10}.}

The focus of our paper is on \fo without any background arithmetical relations. Our interest
in this, weakest setting is motivated
by the prominent problem on the power of encoding-independent computations; see, e.g., \cite{GraedelG15}.
It is a long-standing open question in finite model theory
as to whether there exists a logic capturing polynomial time on
finite relational structures. The existence of a natural logic capturing polynomial time
would mean that any polynomial-time computation could be made, in a sense,
independent of the input encoding.
If this is true, are the encoding-independent computations necessarily
slower than the standard ones?
This question admits the following natural variation.
Suppose that a decision problem
a priori admits an encoding-independent polynomial-time algorithm, say,
being definable in \fo, like \probl{Subgraph Isomorphism} for a fixed pattern graph $F$.
Is it always true that the running time of this algorithm
can be improved in the standard encoding-dependent Turing model of
computation?

A straightforward conversion of an FO sentence defining \subgr F
into an algorithm recognizing \subgr F results in the time bound $O(n^{D(F)})$
for \probl{Subgraph Isomorphism}, which can actually be improved to $O(n^{W(F)})$;
see \cite[Prop.~6.6]{Libkin04}. The same applies to $\fo[<]$.
The last logic is especially interesting in the context of \emph{order-invariant definitions}.
It is well known \cite{Libkin04,Schweikardt13} that there are
properties of (unordered) finite structures that can be defined in $\fo[<]$
but not in \fo. Even if a property, like \subgr F, is definable in FO,
one can expect that in $\fo[<]$ it can be defined much more succinctly.
As a simple example, take $F$ to be the star graph $K_{1,s}$ and
observe that $D_<(K_{1,s})\le\log_2s+3$ and $W_<(K_{1,s})\le3$ while $W(K_{1,s})=s+1$.

The main goal we pose in this paper is examining abilities and limitations
of the ``pure'' FO in succinctly defining \probl{Subgraph Isomorphism}.
Actually, if a pattern graph $F$ has $\ell$ vertices, then
the trivial upper bound $D(F)\le\ell$ cannot be improved. We have $W(F)=\ell$
simply because no first-order formula with less than $\ell$ variables
can distinguish between the complete graphs $K_\ell$ and $K_{\ell-1}$.
Is this, however, the only reason preventing more succinct definitions of \subgr F?
How succinctly can \subgr F be defined on large enough graphs?
The question can be formalized as follows.
We say that a sentence $\Phi$ defines \subgr F on
\emph{sufficiently large connected graphs} if
there is $k$ such that the equivalence \refeq{defines} with $\classc=\subgr F$ is true for all
connected $G$ with at least $k$ vertices.
Let $W_v(F)$ (resp.\ $D_v(F)$) denote the minimum variable width (resp.\ quantifier depth) of such $\Phi$.

Throughout the paper, we assume that the fixed \textbf{pattern graph $F$ is connected}.
Therefore, $F$ is contained in a host graph $G$ if and only if
it is contained in a connected component of $G$. By this reason,
the decision problem for \subgr F efficiently reduces to its restriction
to connected input graphs. Since it suffices to
solve the problem only on all sufficiently large inputs,
\subgr F is still recognizable in time $O(n^{W_v(F)})$,
while $W_v(F)\le W(F)$.

A further relaxation is motivated by Courcelle's theorem \cite{Courcelle90} saying that every graph property
definable by a sentence in monadic second-order logic
can be efficiently decided on graphs of bounded treewidth.
More precisely, for \probl{Subgraph Isomorphism}
Courcelle's theorem implies that \subgr F is decidable in time
$f(\ell,\tw(G))\cdot n$, which means linear time for any class of input graphs
having bounded treewidth.

Now, we say that a sentence $\Phi$ defines \subgr F on
\emph{connected graphs with sufficiently large treewidth} if
there is $k$ such that the equivalence \refeq{defines}  with $\classc=\subgr F$ is true for all
connected $G$ with treewidth at least $k$.
Denote the minimum variable width (resp.\ quantifier depth)
of such $\Phi$ by $W_\tw(F)$ (resp.\ $D_\tw(F)$).
Fix $k$ that ensures the minimum value $W_\tw(F)$ and recall that, by Courcelle's theorem,
the subgraph isomorphism problem is solvable on graphs with treewidth less than $k$
in linear time. Note that, for a fixed $k$, whether or not $\tw(G)<k$
is also decidable in linear time \cite{Bodlaender96}.
It follows that \subgr F is recognizable even in time $O(n^{W_\tw(F)})$,
while $W_\tw(F)\le W_v(F)$.

The above discussion shows that the parameters $W_v(F)$, $D_v(F)$, $W_\tw(F)$, and $D_\tw(F)$ have clear
algorithmic meaning. Analyzing this setting, we obtain the following results.

\begin{itemize}
\item
We demonstrate that non-trivial definitions over sufficiently large graphs are
possible by showing that $D_v(F)\le v(F)-3$ for some $F$, where $v(F)$ denotes
the number of vertices in $F$. On the other hand, we show limitations
of this approach by proving that $W_v(F)\ge(v(F)-1)/2$ for all~$F$.
\item
The last barrier (as well as any lower bound in terms of $v(F)$)
can be overcome by definitions over graphs with sufficiently large treewidth.
Specifically, for every $\ell$ and $a\le\ell$ there is an $\ell$-vertex $F$
such that $D_\tw(F)\le a$ and, moreover, $\tw(F)=a-1$. On the other hand, $W_\tw(F)\ge\tw(F)$ for all $F$.
Note that, along with \refeq{W-tw}, this implies that
$W_\Arb(F)\le W_\tw(F)+3$.
\end{itemize}

Furthermore, we also consider definitions of \subgr F over graphs of
\emph{sufficiently large connectedness}. Denote the corresponding quantifier depth
parameter by $D_\kappa(F)$ and note that $D_\kappa(F)\le D_\tw(F)\le D_v(F)$
(see Section \ref{s:prel} for details), which motivates our interest in lower bound for $D_\kappa(F)$.
Let $e(F)$ denote the number of edges in $F$. For every pattern graph $F$ with $e(F)>v(F)$,
we prove that $D_\kappa(F)\ge\frac{e(F)}{v(F)}+2$.

Finally, we determine the exact values of the parameters $D_\kappa(F)$, $D_\tw(F)$, and $D_v(F)$
for all connected pattern graphs $F$ on 4 vertices.

\paragraph{Related work.}
In an accompanying paper \cite{induced},
we address the descriptive complexity of the \probl{Induced Subgraph Isomorphism} problem.
Let \indsubgr F denote the class of all graphs containing
an \emph{induced} subgraph isomorphic to $F$.
The state-of-the-art of the algorithmics for \probl{Induced Subgraph Isomorphism}
is different from \probl{Subgraph Isomorphism}. Floderus et al.~\cite{FloderusKLL15}
collected evidences in favour of the conjecture that \indsubgr F for $F$ with $\ell$
vertices cannot be recognized faster than \indsubgr{K_{c\,\ell}}, where $c<1$ is a constant.
Similarly to $D(F)$, we use notation $D[F]=D(\indsubgr F)$ and $W[F]=W(\indsubgr F)$,
where the square brackets indicate that the case of induced subgraphs is considered.
The trivial argument showing that $W(F)=v(F)$ does not work anymore unless $F$
is a complete graph. Proving or disproving that $D[F]=W[F]=v(F)$
seems to be a subtle problem. An example of a pattern graph $F$ for which $D[F]<v(F)$
is given by considering the paw graph, as a consequence of Olariu's characterization
of the class of paw-free graphs \cite{Olariu88}. In \cite{induced}, we prove a
general lower bound $W[F]\ge(1/2-o(1))v(F)$.

\paragraph{Organization of the paper.}
We introduce our setting formally in Section \ref{s:prel}, which also
contains necessary logical and graph-theoretic preliminaries.
The first-order definitions of \subgr F over sufficiently large connected graphs
(the parameters $D_v(F)$ and $W_v(F)$) are addressed in Section \ref{s:large}.
Sections \ref{s:large-tw} and \ref{s:conn} are devoted to the definitions over
graphs of sufficienty large treewidth ($D_\tw(F)$ and $W_\tw(F)$) and of
sufficienty large connectedness ($D_\kappa(F)$)
respectively. The exact values of $D_v(F)$, $D_\tw(F)$, and $D_\kappa(F)$ are
determined for all connected $F$ on 4 vertices in Section \ref{s:four}.
The width parameters are also determined with the exception of $W_\kappa(F)$
for $F$ being the diamond graph and the 4-cycle. We conclude with discussing
further questions in Section~\ref{s:que}.

A preliminary version of this paper appeared in~\cite{VZh16}.

\section{Preliminaries}\label{s:prel}

\subsection{First-order complexity of graph properties}

We consider first-order sentences about graphs
in the language containing the adjacency and the equality
relations.
Let \classc be a first-order definable class of graphs and $\pi$ be a graph parameter.
Let $D_\pi^k(\classc)$ denote the minimum quantifier depth of a first-order sentence $\Phi$
such that, for every connected graph $G$ with $\pi(G)\ge k$,
$\Phi$ is true on $G$ exactly when $G$ belongs to \classc.
Note that $D_\pi^{k}(\classc)\ge D_\pi^{k+1}(\classc)$, and define $D_\pi(\classc)=\min_k D_\pi^k(\classc)$.
In other words, $D_\pi(\classc)$ is the minimum quantifier depth of a first-order sentence
defining \classc over connected graphs with sufficiently large values of~$\pi$.

The \emph{variable width} of a first-order sentence $\Phi$ is the number of first-order variables
used to build $\Phi$; different occurrences of the same variable do not count.
Similalrly to the above, by $W_\pi(\classc)$ we denote the minimum variable width of $\Phi$
defining \classc over connected graphs with sufficiently large~$\pi$.
Note that
$$
W_\pi(\classc)\le D_\pi(\classc).
$$

Recall that a graph is \emph{$k$-connected} if it has more than $k$ vertices,
is connected, and remains connected after removal of any $k-1$ vertices.
The \emph{connectivity} $\kappa(G)$ of $G$ is equal to the maximum $k$
such that $G$ is $k$-connected.
We will consider the depth $D_\pi(\classc)$ and the width $W_\pi(\classc)$ for three parameters $\pi$ of a graph $G$, namely
the number of vertices $v(G)$, the treewidth $\tw(G)$, and the connectivity $\kappa(G)$.
Note that $\tw(G)<v(G)$. Note also that any graph $G$ with $v(G)>k$ and $\tw(G)<k$
can be disconnected by removing fewer than $k$ vertices. Therefore,
every $k$-connected graph has treewidth at least $k$. It follows that
$$
D_v^{k+1}(\classc)\ge D_\tw^k(\classc)\ge D_\kappa^k(\classc)
$$
and, hence,
$$
D_v(\classc)\ge D_\tw(\classc)\ge D_\kappa(\classc).
$$
Similarly,
$$
W_v(\classc)\ge W_\tw(\classc)\ge W_\kappa(\classc).
$$
As it was discussed in Section \ref{s:intro}, the values of $D_v(\classc)$ and $D_\tw(\classc)$,
as well as $W_v(\classc)$ and $W_\tw(\classc)$,
are related to the time complexity of the decision problem for \classc.
Consideration of $D_\kappa(\classc)$ and $W_\kappa(\classc)$ is motivated by the fact that some lower bounds
we are able to show for $D_v(\classc)$ and $D_\tw(\classc)$ actually hold for $D_\kappa(\classc)$
or even for $W_\kappa(\classc)$,
and it is natural to present them in this stronger form.

Recall that \subgr F denotes the class of graphs containing a subgraph
isomorphic to $F$. Simplifying the notation, we write
$D_v(F)=D_v(\subgr F)$, $W_v(F)=W_v(\subgr F)$, etc.

Given two non-isomorphic graphs $G$ and $H$, let
$D(G,H)$ (resp.\ $W(G,H)$) denote the minimum quantifier depth (resp.\ variable width) of a sentence
that is true on one of the graphs and false on the other.

\begin{lemma}\label{lem:DD}
\mbox{}

  \begin{enumerate}
\item
$D_\pi(\classc)\ge d$ if there are connected graphs $G\in\classc$ and $H\notin\classc$
with arbitrarily large values of $\pi(G)$ and $\pi(H)$ such that $D(G,H)\ge d$.
\item
$W_\pi(\classc)\ge d$ if there are connected graphs $G\in\classc$ and $H\notin\classc$
with arbitrarily large values of $\pi(G)$ and $\pi(H)$ such that $W(G,H)\ge d$.
\item
$D_\pi(\classc)\le d$ if $D(G,H)\le d$ for all connected graphs $G\in\classc$ and $H\notin\classc$
with sufficiently large values of $\pi(G)$ and~$\pi(H)$.
  \end{enumerate}
\end{lemma}

\begin{proof}
Parts 1 and 2 follow directly from the definitions as any sentence defining \classc
on connected graphs with sufficiently large $\pi$ distinguishes between any two graphs
$G\in\classc$ and $H\notin\classc$ with sufficiently large $\pi$.
Let us prove Part 3. By assumption, any two connected graphs $G\in\classc$ and $H\notin\classc$
with sufficiently large $\pi$ (say, with $\pi(G)\ge k$ and $\pi(H)\ge k$)
are distinguished by a sentence $\Phi_{G,H}$ of quantifier depth at most $d$
(that is true on $G$ and false on $H$).
For a connected graph $G\in\classc$ with $\pi(G)\ge k$,
consider the sentence $\Phi_G\feq\bigwedge_{H}\Phi_{G,H}$,
where the conjunction is over all connected $H\notin\classc$ with $\pi(H)\ge k$.
This sentence distinguishes $G$
from all $H\notin\classc$ with $\pi(H)\ge k$ and has quantifier depth at most $d$.
The only problem with it is that the conjunction over $H$ is actually infinite.
Luckily, there are only finitely many pairwise
inequivalent first-order sentences about graphs
of quantifier depth $d$; see, e.g., \cite[Theorem 2.4]{PikhurkoV11}.
Removing all but one formula $\Phi_{G,H}$ from each equivalence class,
we make $\Phi_G$ a legitimate finite sentence.
Now, consider $\Phi\feq\bigvee_{G}\Phi_{G}$,
where the disjunction is over all connected $G\in\classc$ with $\pi(G)\ge k$.
It can be made finite in the same way.
The sentence $\Phi$ defines $\classc$ over connected graphs with $\pi(G)\ge k$ and has quantifier depth $d$.
Therefore, $D_\pi(\classc)\le D^k_\pi(\classc)\le d$.
\end{proof}

Lemma \ref{lem:DD} reduces estimating $D_\pi(\classc)$ to estimating $D(G,H)$
over connected $G\in\classc$ and $H\notin\classc$ with large values of $\pi$.
Also, proving lower bounds for $W_\pi(\classc)$ reduces to proving lower bounds for $W(G,H)$.
For estimating $D(G,H)$ and $W(G,H)$ there is a remarkable tool.

In the \emph{$k$-pebble Ehrenfeucht-Fra{\"\i}ss{\'e} game},
the board consists of
two vertex-disjoint graphs $G$ and $H$.
Two players, \emph{Spoiler} and \emph{Duplicator} (or \emph{he} and \emph{she})
have equal sets of $k$ pairwise different pebbles.
In each round, Spoiler takes a pebble and puts it on a vertex in $G$ or in $H$;
then Duplicator has to put her copy of this pebble on a vertex
of the other graph.
Duplicator's objective is to ensure that the pebbling determines a partial
isomorphism between $G$ and $H$ after each round; when she fails, she immediately loses.
The proof of the following facts can be found in \cite{Immerman-book}:
\begin{enumerate}
\item
 $D(G,H)$ is equal to the
minimum $k$ such that Spoiler has a winning strategy in the $k$-round $k$-pebble
game on $G$ and~$H$.
\item
 $W(G,H)$ is equal to the
minimum $k$ such that, for some $d$, Spoiler has a winning strategy in the $d$-round $k$-pebble
game on $G$ and~$H$.
\end{enumerate}

\subsection{Graph-theoretic preliminaries}

Recall that $v(G)$ denotes the number of vertices in a graph $G$.
The treewidth of $G$ is denoted by~$\tw(G)$.
The \emph{neighborhood} $N(v)$ of a vertex $v$ consists of
all vertices adjacent to $v$. The number $\deg v=|N(v)|$
is called the \emph{degree} of $v$.
The vertex of degree 1 is called \emph{pendant}.

We use the standard notation
$K_n$ for complete graphs,
$P_n$ for paths, and $C_n$ for cycles on $n$ vertices.
Futhermore, $K_{a,b}$ denotes the complete bipartite graph whose
vertex classes have $a$ and $b$ vertices. In particular, $K_{1,n-1}$
is the star graph on $n$ vertices.
The subscript in the name of a graph will almost always denote the number of vertices.
If a graph is indexed by two parameters, their sum is typically equal to the total
number of vertices in the graph.

\begin{figure}
  \centering
\begin{tikzpicture}[every node/.style={circle,draw,inner sep=2pt,fill=Black},scale=.5]
  \begin{scope}
    \path (0,0) node (a1) {}
      (0,1) node (a2) {} edge (a1)
      (0,2) node (a3) {} edge (a2)
      (0,3) node (a4) {} edge (a2)
   (-1,2.5) node (a5) {} edge (a4) edge (a3)
    (1,2.5) node (a6) {} edge (a5) edge (a4) edge (a3);
\node[draw=none,fill=none] at (0,-1) {$L_{4,2}$};
  \end{scope}

  \begin{scope}[xshift=40mm]
    \path (0,0) node (a1) {}
      (0,1) node (a2) {} edge (a1)
      (0,2) node (a3) {} edge (a2)
      (0,3) node (a4) {} edge (a2)
   (-1,2.5) node (a5) {} edge (a3)
    (1,2.5) node (a6) {} edge (a3);
\node[draw=none,fill=none] at (0,-1) {$S_{4,2}$};
  \end{scope}

  \begin{scope}[xshift=80mm]
    \path
      (0,1) node (a2) {}
      (0,2) node (a3) {} edge (a2)
      (0,3) node (a4) {} edge (a2)
   (-1,2.5) node (a5) {} edge (a4) edge (a3)
    (1,2.5) node (a6) {} edge (a5) edge (a4) edge (a3)
  (-.7,1.3) node (a1) {} edge (a3)
   (.7,1.3) node (a7) {} edge (a3);
\node[draw=none,fill=none] at (0,-1) {$J_{4,3}$};
  \end{scope}

  \begin{scope}[xshift=120mm]
    \path[scale=.8,yshift=7.5mm]
      (0,1) node (x) {}
      (0,2) node (a1) {} edge (x)
      (0,3) node (a2) {} edge (a1)
    (-1,.3) node (b1) {} edge (x)
   (-2,-.4) node (b2) {} edge (b1)
     (1,.3) node (c1) {} edge (x)
    (2,-.4) node (c2) {} edge (c1);
\node[draw=none,fill=none] at (0,-1) {$M_{3,2}$};
  \end{scope}
\end{tikzpicture}
  \caption{Special graph families: Lollipops, sparklers, jellyfishes, and megastars.}
  \label{fig:special-graphs}
\end{figure}
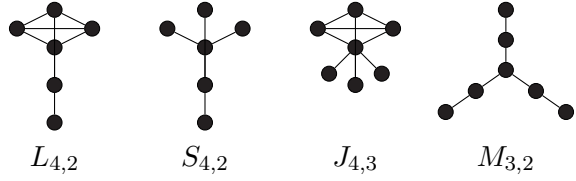

The following definitions are illustrated in Fig.~\ref{fig:special-graphs}.
Let $a\ge3$ and $b\ge1$. The \emph{lollipop graph} $L_{a,b}$ is obtained from $K_a$ and $P_b$
by adding an edge between an end vertex of $P_b$ and a vertex of $K_a$.
We also make a natural convention that $L_{a,0}=K_a$.
Furthermore, the \emph{sparkler graph} $S_{a,b}$ is obtained from $K_{1,a-1}$ and $P_b$
by adding an edge between an end vertex of $P_b$ and the central vertex of~$K_{1,a-1}$.
The \emph{jellyfish graph} $J_{a,b}$ is the result of attaching $b$ pendant vertices
to a vertex of $K_a$. Finally,
the \emph{megastar} graph $M_{s,t}$ is obtained from the star $K_{1,s}$
by subdividing each edge into $P_{t+1}$; thus $v(M_{s,t})=st+1$.

\section{Definitions over sufficiently large graphs}\label{s:large}

Our first goal is to demonstrate that non-trivial definitions over large connected graphs
are really possible. The lollipop graphs $L_{a,1}$ give simple examples of pattern graphs $F$ with $D_v(F)\le v(F)-1$.
Though not so easily, the same can be shown
for the path graphs $P_\ell$. We are able to show better upper bounds using sparkler graphs.

\begin{theorem}\label{thm:S44}
There is a graph $F$ with $D_v(F)\le v(F)-3$. Specifically,
$D_v(S_{4,4})=5$.
\end{theorem}

For the proof we need two technical lemmas.

\begin{lemma}\label{lem:large-deg}
  Suppose that a connected graph $H$ contains the 4-star $K_{1,4}$ as a subgraph
but does not contain any subgraph $S_{4,4}$. Then $H$ contains a vertex
of degree more than $(v(H)/2)^{1/7}$.
\end{lemma}

\begin{proof}
$H$ cannot contain $P_{15}$ because, together with $K_{1,4}$,
it would give an $S_{4,4}$ subgraph. Consider an arbitrary spanning tree $T$ in $H$
and denote its maximum vertex degree by $d$ and its radius by $r$.
Note that $v(T)\le 1+d+d(d-1)+\ldots+d(d-1)^{r-1}$.
Since $T$ contains no $P_{15}$, we have $r\le 7$. It follows that $v(H)=v(T)<2d^{7}$.
\end{proof}

Let $\sim$ denote the adjacency relation and recall that $N(v)$ denotes the
neighborhood of a vertex~$v$.

\begin{lemma}\label{lem:S44}
Let $y_0\in V(H)$ and assume that
\begin{itemize}
\item $H$ is a sufficiently large connected graph,
\item $H$ does not contain $S_{4,4}$,
\item $\deg y_0\geq 4$,
\item $y_0y_1y_2y_3y_4$ is a path in $H$.
\end{itemize}
Then (see Fig.~\ref{fig:proof:S44})
\begin{enumerate}
\item $\deg y_0=4$,
\item $y_0\sim y_2$, $y_0\nsim y_3$, $y_0\nsim y_4$,
\item if $N(y_0)=\{y_1,y_2,y',y''\}$, then $y_1\nsim y'$ and $y_1\nsim y''$.
\end{enumerate}
\label{44}
\end{lemma}

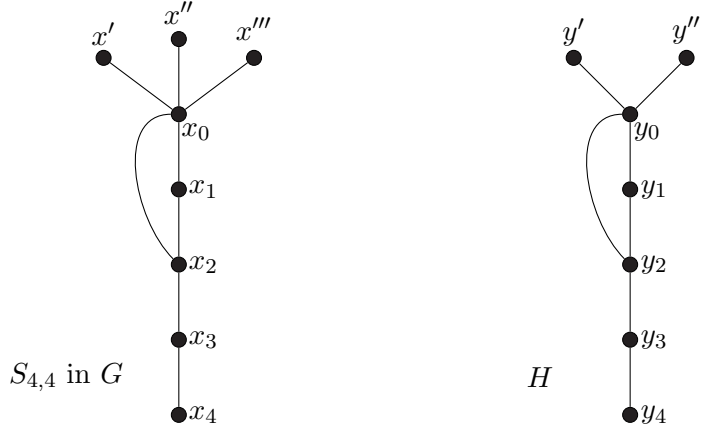
\begin{figure}
  \centering
\begin{tikzpicture}[every node/.style={circle,draw,inner sep=2pt,fill=Black}]
  \begin{scope}
\path (0,0) node (x0) {}
       (0,-1) node (x1) {} edge (x0)
       (0,-2) node (x2) {} edge (x1) edge[out=135, in=180] (x0)
       (0,-3) node (x3) {} edge (x2)
       (0,-4) node (x4) {} edge (x3)
       (0,1) node (x'') {} edge (x0)
       (-1,.75) node (x') {} edge (x0)
       (1,.75) node (x''') {} edge (x0);
\node[draw=none,fill=none,below right] at (x0) {$x_0$};
\node[draw=none,fill=none,right] at (x1) {$x_1$};
\node[draw=none,fill=none,right] at (x2) {$x_2$};
\node[draw=none,fill=none,right] at (x3) {$x_3$};
\node[draw=none,fill=none,right] at (x4) {$x_4$};
\node[draw=none,fill=none,above] at (x') {$x'$};
\node[draw=none,fill=none,above] at (x'') {$x''$};
\node[draw=none,fill=none,above] at (x''') {$x'''$};
\node[draw=none,fill=none] at (-1.5,-3.5) {$S_{4,4}$ in $G$};
  \end{scope}
  \begin{scope}[xshift=60mm]
\path (0,0) node (x0) {}
       (0,-1) node (x1) {} edge (x0)
       (0,-2) node (x2) {} edge (x1) edge[out=135, in=180] (x0)
       (0,-3) node (x3) {} edge (x2)
       (0,-4) node (x4) {} edge (x3)
       (-.75,.75) node (x') {} edge (x0)
       (.75,.75) node (x'') {} edge (x0);
\node[draw=none,fill=none,below right] at (x0) {$y_0$};
\node[draw=none,fill=none,right] at (x1) {$y_1$};
\node[draw=none,fill=none,right] at (x2) {$y_2$};
\node[draw=none,fill=none,right] at (x3) {$y_3$};
\node[draw=none,fill=none,right] at (x4) {$y_4$};
\node[draw=none,fill=none,above] at (x') {$y'$};
\node[draw=none,fill=none,above] at (x'') {$y''$};
\node[draw=none,fill=none] at (-1.2,-3.5) {$H$};
  \end{scope}
\end{tikzpicture}
  \caption{Proof of Theorem \ref{thm:S44}.}
  \label{fig:proof:S44}
\end{figure}

\begin{proof}
By Lemma \ref{lem:large-deg} we know that $H$ must contain a vertex $z$ of large degree,
namely $\deg z\geq 7$. We have $y_0\nsim y_4$ for else $H$ would contain a cycle $C_5$ and, together with $z$,
this would give us a subgraph $S_{4,4}$ (because, by connectedness of $H$, we would have a path $P_5$ emanating from $z$).
Therefore, $y_0$ has a neighbor $y'\notin\{y_1,y_2,y_3,y_4\}$. Furthermore,
$y_0\nsim y_3$ for else, considering a path from $z$ to one of the
vertices $y',y_0,y_1,y_2,y_3,y_4$, we get a $P_5$ emanating from $z$ and, hence, an $S_{4,4}$.
Therefore, $y_0$ has another neighbor $y''\notin\{y',y_1,y_2,y_3,y_4\}$.  Furthermore,
$y_0\sim y_2$ for else $y_0$ would have three neighbors $y',y'',y'''$ different from $y_1,y_2,y_3,y_4$,
which would give $S_{4,4}$. By the same reason, $y_0$ has no other neighbors, that is,
$N(y_0)=\{y_1,y_2,y',y''\}$ and $\deg y_0=4$. Note that $z\in\{y_0,y_1,y_2,y_3,y_4\}$
for else we easily get an $S_{4,4}$ by considering a path from $z$ to one of these vertices.
It is also easy to see that $z\neq y_0,y_4,y_3,y_1$
(for example, if $\deg y_1\geq 7$, then it would give an $S_{4,4}$ with tail $y_1y_0y_2y_3y_4$).
Therefore, $z=y_2$. If $y_1\sim y'$ or $y_1\sim y''$, we would have an $S_{4,4}$
with tails $y_2y_1y'y_0y''$ or $y_2y_1y''y_0y'$ respectively.
\end{proof}

\begin{proof}[Proof of Theorem \ref{thm:S44}]
We are now ready to prove the upper bound $D_v(S_{4,4})\le5$.
Consider sufficiently large connected graphs $G$ and $H$ and
suppose that $G$ contains an $S_{4,4}$, whose vertices are labeled as in Fig.~\ref{fig:proof:S44},
and $H$ contains no copy of $S_{4,4}$.
We describe a winning strategy for Spoiler in the game on $G$ and~$H$.

{\it 1st round}. Spoiler pebbles $x_0$. Denote the response of Duplicator in $H$ by $y_0$.
Assume that $\deg y_0\geq 4$ for else Spoiler
wins in the next $4$ moves. Assume that $x_0\sim x_2$ for else Spoiler wins by pebbling $x_1,x_2,x_3,x_4$
(if Duplicator responds with a path $y_0y_1y_2y_3y_4$, she loses by Condition 2 in Lemma~\ref{44}).

{\it 2nd round}. Spoiler pebbles $x_1$. Denote the response of Duplicator in $H$ by $y_1$.
Assume that there is a path $y_0y_1y_2y_3y_4$ for else Spoiler wins in the next $3$ moves.

\Case 1{$x_1$ is adjacent to any of the vertices $x',x'',x'''$, say, to $x'$.}
Spoiler pebbles $x_2$ and $x'$ and wins. Indeed, Duplicator has to respond with two vertices in $H$
both in $N(y_0)\cap N(y_1)$, which is impossible by Conditions 1 and 3 of Lemma~\ref{44}.

\Case 2{$x_1\nsim x'$, $x_1\nsim x''$, $x_1\nsim x'''$.}
Spoiler wins by pebbling $x',x'',x'''$. Duplicator has to respond with three vertices
in $N(y_0)\setminus N(y_1)$, which is impossible by Conditions 1 and 2 of Lemma~\ref{44}.

This completes the proof of the upper bound. On the other hand, we have
$D_v(S_{4,4})>4$ by considering the jellyfish graphs $G=J_{5,n}$ and $H=J_{4,n}$.
\end{proof}

With more technical effort, we can show that $D_v(F)\le v(F)-3$ for infinitely many $F$,
namely for all $F=S_{a,a}$ (the proof of this fact is rather involved and will appear elsewhere).

We now show general lower bounds for $D_v(F)$ and $W_v(F)$.
For this, we need some definitions.
Let $v_0v_1\ldots v_t$ be an induced path in a graph $G$.
We call it \emph{pendant} if $\deg v_0\ne2$, $\deg v_t=1$ and $\deg v_i=2$ for all $1\le i<t$.
Furthermore, let $S$ be an induced star $K_{1,s}$ in $G$
with the central vertex $v_0$. We call $S$ \emph{pendant}
if all its pendant vertices are pendant also in $G$,
and in $G$ there is no more than $s$ pendant vertices adjacent to $v_0$.
The definition ensures that a pendant path (or star) cannot be
contained in a larger pendant path (or star). As an example,
note that the sparkler graph $S_{s+1,t}$ has a pendant $P_{t+1}$
and a pendant~$K_{1,s}$.

Let $p(F)$ denote the maximum $t$ such that $F$ has a pendant path~$P_{t+1}$.
Similarly, let $s(F)$ denote the maximum $s$ such that $F$ has a pendant star~$K_{1,s}$.
If $F$ has no pendant vertex, then we set $p(F)=0$ and $s(F)=0$.

\begin{theorem}\label{thm:ell/2}
$D_v(F)\ge(v(F)+1)/2$ and $W_v(F)\ge(v(F)-1)/2$ for every connected $F$ unless $F=P_2$ or $F=P_3$.
\end{theorem}

\begin{proof}
Denote
$$
\ell=v(F),\ t=p(F)\text{ and }s=s(F).
$$
We begin with noticing that
\begin{equation}\label{eq:pendant-path}
D_v(F)\ge\ell-t \text{ and } W_v(F)\ge\ell-t-1.
\end{equation}
Indeed, this is obvious if $F$ is a path, that is, $F=P_{t+1}$. If $F$ is not a path,
we consider lollipop graphs $G=L_{\ell-t,n}$ and $H=L_{\ell-t-1,n}$ for each $n\ge t$
(note that $\ell\ge t+3$ and, if $\ell=t+3$, then $H=L_{2,n}=P_{n+2}$). Obviously,
$G$ contains $F$, and $H$ does not. It remains to note that $D(G,H)\ge\ell-t$ and $W(G,H)\ge\ell-t-1$.

We also claim that
\begin{equation}\label{eq:pendant-star}
D_v(F)\ge\ell-s \text{ and } W_v(F)\ge\ell-s-1.
\end{equation}
This is obvious if $F$ is a star, that is, $F=K_{1,s}$. If $F$ is not a star,
we consider jellyfish graphs $G=J_{\ell-s,n}$ and $H=J_{\ell-s-1,n}$ for each $n\ge s$
(note that $\ell\ge s+3$ and, if $\ell=s+3$, then $H=J_{2,n}=K_{1,n+1}$).
Clearly, $G$ contains $F$, and $H$ does not. It remains
to observe that $D(G,H)\ge\ell-s$ and $W(G,H)\ge\ell-s-1$.

Let $F=K_{1,\ell-1}$ or $F=P_\ell$, where $\ell\ge4$.
Using \refeq{pendant-path} and \refeq{pendant-star} respectively,
we get $D_v(F)\ge\ell-1\ge\frac{\ell+1}2$ and, similarly, $W_v(F)\ge\ell-2\ge\frac{\ell-1}2$.
Assume, therefore, that $F$ is neither a star nor a path.
In this case we claim that
\begin{equation}
  \label{eq:stell}
t+s<\ell.
\end{equation}
This is obviously true if $F$ has no pendant vertex, that is, $t=s=0$.
Suppose that $F$ has a pendant vertex and, therefore,
both $t>0$ and $s>0$. Consider an arbitrary spanning tree $T$ of $F$ and
note that $T$ contains all pendant paths and stars of $F$.
Fix a longest pendant path $P$ and a largest pendant star $S$ in $F$.
If $P$ and $S$ share at most one common vertex, we readily get \refeq{stell}.
If they share two vertices, then $S=K_{1,1}$, i.e., $s=1$,
and $t+1<\ell$ follows from the assumption that $F$ is not a path.

The theorem readily follows from \refeq{pendant-path}, \refeq{pendant-star}, and~\refeq{stell}.
\end{proof}

\section{Definitions over graphs of sufficiently large treewidth}\label{s:large-tw}

Theorem \ref{thm:ell/2} poses limitations on the succinctness
of definitions over sufficiently large connected graphs.
We now show that there are no such limitations for definitions
over connected graphs with sufficiently large treewidth.

The Grid Minor Theorem says that
every graph of large treewidth contains a large grid minor; see \cite{Diestel}.
The strongest version of this result belongs to Chekuri and Chuzhoy \cite{ChekuriC14}
who proved that, for some $\epsilon>0$, every graph $G$ of treewidth $k$
contains the $m\times m$ grid as a minor with $m=\Omega(k^\epsilon)$.
If $m>2b$, then $G$ must contain $M_{3,b}$ as a subgraph.
This applies also to all subgraphs of $M_{3,b}$. The following
result is based on the fact that a graph of large treewidth contains a long path.

\begin{theorem}\label{thm:L_ab}
For all $a$ and $\ell$ such that $3\le a\le\ell$ there is a graph $F$ with
$v(F)=\ell$ and $\tw(F)=a-1$ such that $D_\tw(F)\le a$. Specifically,
  $D_\tw(L_{a,b})=W_\kappa(L_{a,b})=a$ if $a\ge3$ and $b\ge0$.
\end{theorem}

Note for comparison that $W_v(L_{a,b})\ge a+b-2$, as follows from the bound \refeq{pendant-star}
in the proof of Theorem~\ref{thm:ell/2}.

\begin{proof}
We first prove the upper bound $D_\tw({L_{a,b}})\le a$.
If a connected graph $H$ of large treewidth does not contain $L_{a,b}$, it
cannot contain even $K_a$ for else $K_a$ could be combined
with a long path to give $L_{a,b}$.
Therefore, Spoiler wins on $G\in\subgr{L_{a,b}}$
and such $H$ in $a$ moves.

For the lower bound $W_\kappa({L_{a,b}})\ge a$,
consider $G=K(a,n)$ and $H=K(a-1,n)$, where
$K(a,n)$ denotes the complete $a$-partite
graph with each part having $n$ vertices. Note that the graph $K(a,n)$ is $(a-1)n$-connected.
If $n>b$, then $G$ contains $L_{a,b}$, while $H$ for any $n$ does not
contain even $K_a$. It remains to note that
$W(G,H)\ge a$ if $n\ge a-1$.
\end{proof}

We now prove a general lower bound for $W_\tw(F)$ in terms of the treewidth~$\tw(F)$.

Using the terminology of \cite[Chap.~5]{JLR-book}, we define the \emph{core} $F_0$
of $F$ to be the graph obtained from $F$ by removing, consecutively and as
long as possible, vertices of degree at most 1. If $F$ is not a forest,
then $F_0$ is nonempty; it consists of all cycles of $F$ and the paths
between them.

We will use the well-known fact that there are cubic graphs of arbirary large treewidth.
This fact dates back to Pinsker \cite{Pinsker73} who showed that a random cubic graph
with high probability has good expansion properties, implying linear treewidth.

\begin{theorem}\label{thm:pendant-path-tw}
\mbox{}

\begin{enumerate}
\item
$W_\tw(F)\ge v(F_0)$ for every~$F$.
\item
$W_\tw(F)\ge\tw(F)+1$ for every connected $F$ except for the case that
$F$ is a subtree of some 3-megastar~$M_{3,b}$.
\end{enumerate}
\end{theorem}

\noindent
Note that the bound in part 2 of Theorem \ref{thm:pendant-path-tw} is tight by Theorem~\ref{thm:L_ab}.

\begin{proof}
\textit{1.}
Denote $v(F)=\ell$ and $v(F_0)=\ell_0$.
If $F$ is a forest, then $\ell_0=0$, and the claim is trivial.
Suppose, therefore, that $F$ is not a forest. In this case, $\ell_0\ge3$.

We begin with a cubic graph $B$ of as large treewidth $\tw(B)$ as desired.
Let $(B)_\ell$ denote the graph obtained from $B$ by subdividing each edge by $\ell$ new vertices.
Since $B$ is a minor of $(B)_\ell$, we have $\tw((B)_\ell)\ge\tw(B)$; see~\cite{Diestel}.

Next, we construct a gadget graph $A$ as follows. By a \emph{$k$-uniform tree}
we mean a tree of even diameter where every non-leaf vertex has degree $k$ and all distances between
a leaf and the central vertex are equal. The graph $A$ is obtained by merging the $\ell$-uniform tree
of radius $\ell$ and $(B)_\ell$; merging is done by identifying one leaf of the tree and
one vertex of~$(B)_\ell$.

We now construct $G$ by attaching a copy of $A$ to each vertex of $K_{\ell_0}$.
Specifically, a copy $A_u$ of $A$ is created for each
vertex $u$ of $K_{\ell_0}$, and $u$ is identified with the central vertex of (the tree part of) $A_u$.
Let $H$ be obtained from $G$ by shrinking its clique part to $K_{\ell_0-1}$.
Since both $G$ and $H$ contain copies of $(B)_\ell$, these two graphs have
treewidth at least as large as $\tw(B)$.

The clique part of $G$ is large enough to host the core $F_0$, and the remaining tree shoots of $F$
fit into the $A$-parts of $G$. Therefore, $G$ contains $F$ as a subgraph.
On the other hand, the clique part of $H$ is too small for hosting $F_0$,
and no cycle of $F$ fits into any $A$-part because $A$ has larger girth than $F$.
Therefore, $H$ does not contain $F$.
It remains to notice that $W(G,H)\ge\ell_0$.

\textit{2.}
Suppose first that a connected graph $F$ is not a tree.
By part 1, we then have
$$
W_\tw(F)\ge v(F_0)\ge\tw(F_0)+1=\tw(F)+1.
$$
If $F$ is a tree and is not contained in any 3-megastar, that is,
has a vertex of degree more than 3 or at least two vertices of degree 3,
then there are connected graphs of arbitrarily large treewidth
that do not contain $F$ as a subgraph (for example, consider $(B)_\ell$ for a connected cubic
graph $B$ as in part 1). Trivially, there are also connected graphs of arbitrarily large treewidth
that contain $F$ as a subgraph. Since one pebble is not enough for Spoiler
to distinguish the latter from the former, we have
$W_\tw(F)\ge2=\tw(F)+1$ in this case.
\end{proof}

Theorem \ref{thm:pendant-path-tw} implies that $W_\tw(F)\ge\tw(F)$
for all $F$. Combining it with
the bound $W_\Arb(F)\le\tw(F)+3$ mentioned in Section \ref{s:intro},
we obtain the following relation.

\begin{corollary}\label{cor:arb-tw}
$W_\Arb(F)\le W_\tw(F)+3$.
\end{corollary}

Note that $W_\Arb(F)$ and $W_\tw(F)$ are within a constant factor from each other for infinitely many $F$.
This is so for $F=K_\ell$ as  $W_\Arb(K_\ell)>\ell/4$ (Rossman \cite{Rossman08}).
On the other hand, a gap between the two parameters can be large.
For example, part 1 of Theorem \ref{thm:pendant-path-tw}
gives $W_\tw(C_\ell)=\ell$ whereas $W_\Arb(C_\ell)\le5$.

\section{Definitions over highly connected graphs}\label{s:conn}

In this section we prove a lower bound for $D_\kappa(F)$
in terms of the density of $F$.
The proof will use known facts about random graphs in the Erd\H{o}s-R\'enyi model $G(n,p)$,
collected below.

The \emph{density} of a graph $K$ is defined to be the ratio $\rho(K)=e(K)/v(K)$.
The maximum $\rho(K)$ over all subgraphs $K$ of a graph $F$ will be denoted by~$\rho^*(F)$.
The following fact from the random graph theory was used also in \cite{LiRR14}
for proving average-case lower bounds on the \ac0 complexity of \probl{Subgraph Isomorphism}.
 \emph{With high probability} means the probability approaching $1$ as $n\to\infty$.

\begin{lemma}[Subgraph Threshold, see {\cite[Chap.~3]{JLR-book}}]\label{lem:threshold}
If $\alpha=1/\rho^*(F)$, then the probability that $G(n,n^{-\alpha})$ contains $F$ as a subgraph
converges to a limit different from $0$ and~$1$ as $n\to\infty$.
\end{lemma}

Let $\alpha>0$. Given a graph $S$ and its subgraph $K$,
we define $f_{\alpha}(S,K)=v(S)-v(K)-\alpha(e(S)-e(K))$.

\begin{lemma}[Generic Extension, see  {\cite[Chap.~10]{AlonS16}}]\label{lem:gen-ext}
Let $F$ be a graph with vertices $v_1,\ldots,v_\ell$ and $K$ be a subgraph of $F$
with vertices $v_1,\ldots,v_k$. Assume that $f_{\alpha}(S,K)>0$ for every subgraph $S$ of $F$
containing $K$ as a proper subgraph. Then with high probability
every sequence of pairwise distinct vertices $x_1,\ldots,x_k$ in $G(n,n^{-\alpha})$
can be extended with pairwise distinct $x_{k+1},\ldots,x_\ell$ such that
$x_i\sim x_j$ if and only if $v_i\sim v_j$ for all $i\le\ell$ and $k<j\le\ell$.
\end{lemma}

\begin{lemma}[Zero-One $d$-Law \cite{Zhuk}]\label{lem:d-law}
Let $0<\alpha<\frac1{d-2}$, and $\Psi$ be a first-order statement of quantifier depth $d$.
Then the probability that $\Psi$ is true on $G(n,n^{-\alpha})$ converges either to $0$
or to $1$ as $n\to\infty$.
\end{lemma}

We are now ready to prove our result.

\begin{theorem}\label{thm:e/v}
If $e(F)>v(F)$, then $D_\kappa(F)\ge\frac{e(F)}{v(F)}+2$.
\end{theorem}

\begin{proof}
By Lemma \ref{lem:threshold}, $\lim_{n\to\infty}\prob{\rgraph\in\subgr F}$ exists
and equals neither $0$ nor $1$.
Assume that a first-order sentence $\Phi$ of quantifier depth $d$
defines $\subgr F$ over $k$-connected graphs for all $k\ge k_0$.
We have to prove that $d\ge\frac{e(F)}{v(F)}+2$, whatever~$k_0$.

By the assumption of the theorem, $\rho^*(F)\ge\rho(F)>1$.
Fix $k$ such that $1+1/k<\rho(F)$ and $k\ge k_0$.
Lemma \ref{lem:gen-ext} implies
that with high probability every two vertices in $\rgraph$ can be connected
by $k$ vertex-disjoint paths (of length $k$ each). Therefore, $\rgraph$ is $k$-connected with high probability.

Since $\Phi$ correctly decides the existence of a subgraph $F$ on all $k$-connected graphs,
$$
\prob{\rgraph\models\Phi}=\prob{\rgraph\in\subgr F}+o(1).
$$
Therefore, $\prob{\rgraph\models\Phi}$ converges to the same limit as $\prob{\rgraph\in\subgr F}$,
which is different from $0$ and $1$. By Lemma \ref{lem:d-law}, this implies that
$\alpha\ge\frac1{d-2}$. From here we conclude that
$$
d\ge\rho^*(F)+2\ge\frac{e(F)}{v(F)}+2,
$$
as required.
\end{proof}

\section{Small pattern graphs}\label{s:four}

Now we aim at determining \emph{exact} values of the depth and the width parameters
for small connected pattern graphs. There are only two connected graphs with 3
vertices, the path graph $P_3$ and the complete graph $K_3$.
Since every connected graph with at least 3 vertices contains a subgraph $P_3$,
we have $D_v(P_3)=1$. Theorem \ref{thm:L_ab} in the case of $L_{a,0}=K_a$ gives us
\begin{equation}
  \label{eq:compl}
W_\kappa(K_a)=a
\end{equation}
for every $a\ge3$. In particular, $W_\kappa(K_3)=3$.

In this section, we consider the six connected graphs with 4 vertices, that are shown
in Table~\ref{fig:small}.
Each patern graph $F$ is presented in this table by a row
consisting of two layers, the upper for the depth parameters and the lower for the width
parameters. Note that the values within each row are monotonically non-decreasing
in the right and right upward directions. To improve visual clarity of the table,
we remove all entries whose values are implied by monotonicity.

\begin{table}[h!]
\begin{center}
\begin{tabular}{||cc|c||%
l@{\extracolsep{-8mm}}r|@{\extracolsep{1mm}}l@{\extracolsep{-8mm}}r|@{\extracolsep{1mm}}l@{\extracolsep{-8mm}}r||}
\hline
\multicolumn{3}{||c||}{\multirow{2}{*}{$F$}}
                     &              & $D_\kappa(F)$ &            & $D_\tw(F)$ &          & $D_v(F)$\\
\multicolumn{3}{||c||}{}
                     & $W_\kappa(F)$ &              &  $W_\tw(F)$ &           & $W_v(F)$ &         \\
\hline\hline
\multirow{2}{*}{
\begin{tikzpicture}[every node/.style=vertex,scale=.4]
\path (0,0) node (a) {}
      (0,1) node (b)  {} edge (a)
      (1,0) node (c)  {} edge (b)
      (1,1) node (d)  {} edge (c);
\end{tikzpicture}
}&\multirow{2}{*}{(path)}&
\multirow{2}{*}{$P_4$}&             &              &            &  1        &          & 3   \\
&&
                      &             &              &            &           &  2       &     \\
\hline
\multirow{2}{*}{
\begin{tikzpicture}[every node/.style=vertex,scale=.4]
\path (0,0) node (a) {}
      (-.5,1) node (b)  {} edge (a)
      (0,1) node (c)  {} edge (a)
      (.5,1) node (d)  {} edge (a);
\end{tikzpicture}
}&\multirow{2}{*}{(claw)}&
\multirow{2}{*}{$K_{1,3}$}&         &              &            &  1        &           &  4  \\
&&
                     &             &              &            &           &  3        &     \\
\hline
\multirow{2}{*}{
\begin{tikzpicture}[every node/.style=vertex,scale=.4]
\path (0,0) node (a) {}
      (0,1) node (b)  {} edge (a)
      (.5,.5) node (c)  {} edge (a) edge (b)
      (1.2,.5) node (d)  {} edge (c);
\end{tikzpicture}
}&\multirow{2}{*}{(paw)}&
\multirow{2}{*}{$L_{3,1}$}&         &              &            &           &          &  3   \\
&&
                     & 3           &              &            &           &          &      \\
\hline
\multirow{2}{*}{
\begin{tikzpicture}[every node/.style=vertex,scale=.4]
\path (0,0) node (a) {}
      (0,1) node (b)  {} edge (a)
      (1,1) node (c)  {} edge (b)
      (1,0) node (d)  {} edge (c) edge (a);
\end{tikzpicture}
}&\multirow{2}{*}{(cycle)}&
\multirow{2}{*}{$C_4$}&            &   4          &            &           &          &      \\
&&
                     & $\ge3$      &              & 4          &           &          &      \\
\hline
\multirow{2}{*}{
\begin{tikzpicture}[every node/.style=vertex,scale=.45]
\path (0,0) node (a) {}
      (0,1) node (b)  {}
      (-.5,.5) node (c)  {} edge (a) edge (b)
      (.5,.5) node (d)  {} edge (a) edge (b) edge (c);
\end{tikzpicture}
}&\multirow{2}{*}{(diamond)}&
\multirow{2}{*}{$K_4\setminus e$}& &  4           &            &           &          &      \\
&&
                     & $\ge3$      &              & 4          &           &          &      \\
\hline
\multirow{2}{*}{
\begin{tikzpicture}[every node/.style=vertex,scale=.35]
\path (0,0) node (a) {}
      (90:1cm) node (b)  {} edge (a)
      (210:1cm) node (c)  {} edge (a) edge (b)
      (-30:1cm) node (d)  {} edge (a) edge (b) edge (c);
\end{tikzpicture}
}&\multirow{2}{*}{(complete)}&
\multirow{2}{*}{$K_4$}&            &              &            &           &          &      \\
&&
                     &  4          &              &            &           &          &      \\
\hline
\end{tabular}
\end{center}
\vspace*{-5mm}
\caption{Results on 4-vertex subgraphs.}\label{fig:small}
\end{table}

We begin with several simple observations. First, if a connected graph $H$
has $n>3$ vertices and does not contain $P_4$, then $H=K_{1,n-1}$.
Second, if a connected graph $H$
has $n$ vertices and does not contain $K_{1,3}$, that is,
the maximum vertex degree of $H$ does not exceed 2, then $H=P_n$ or $H=C_n$.
In each case, $H$ has treewidth at most 2. It readily follows that
$D_\tw(P_4)=1$ and $D_\tw(K_{1,3})=1$.

Moreover, the simple structure of connected $K_{1,3}$-free graphs easily implies
that $W_v(K_{1,3})\allowbreak=3$ and $D_v(K_{1,3})=4$.
The lower bounds follow here from the inequalities
$W(M_{3,b},P_n)>2$ and $D(M_{3,b},P_n)>3$ that are true for all $b\ge2$ and $n\ge5$.
To see the upper bound $W_v(K_{1,3})\le3$, consider the 3-pebble Ehrenfeucht-Fra{\"\i}ss{\'e} game
on $G$ with a vertex of degree at least 3 and $H=P_n$ or $H=C_n$ with $n>6$.
In the first three rounds, Spoiler pebbles three vertices in $G$ having a common neighbor.
If Duplicator is still alive then, whatever she responds, two of her vertices in $H$
are at the distance more than 2. This allows Spoiler to win in the next round.

By \refeq{compl}, we have $W_\kappa(K_4)=4$.
As another direct consequence of Theorem \ref{thm:L_ab},
$W_\kappa(L_{3,1})=D_v(L_{3,1})=3$. Here, the upper bound $D_v(L_{3,1})\le3$
follows from the observation that a connected graph with at least 4 vertices
contains a subgraph $L_{3,1}$ if and only if it contains a subgraph~$K_3$.

If $F$ is 2-connected,
part 1 of Theorem \ref{thm:pendant-path-tw} implies that $W_\tw(F)=v(F)$.
This applies to the 4-cycle and the diamond graph, and we have
the equalities $W_\tw(C_4)=4$ and $W_\tw(K_4\setminus e)=4$.
Furthermore, $D_\kappa(K_4\setminus e)=4$ as a consequence of Theorem~\ref{thm:e/v}.
The lower bound $W_\kappa(K_4\setminus e)\ge3$ can be seen by considering,
like in the proof of Theorem \ref{thm:L_ab}, the complete multipartite graphs
$G=K(3,n)$ and $H=K(2,n)$.

The other entries of Table~\ref{fig:small} are not so obvious.

\begin{theorem}\label{thm:four}
\mbox{}

\begin{enumerate}
\item
$W_v(P_4)=2$ and $D_v(P_4)=3$.
\item
$D_\kappa(C_4)=4$ and $W_\kappa(C_4)\ge3$.
\end{enumerate}
\end{theorem}

The proof of Theorem \ref{thm:four} takes the rest of this section.

\paragraph{The path subgraph ($P_4$).}

We restate part 1 of Theorem \ref{thm:four} as two lemmas.

\begin{lemma}
$D_v({P_4})=3$.
\end{lemma}

\begin{proof}
We first show that $D^4_v({P_4})\le3$.
 Indeed, the list of connected graphs not containing $P_4$
consists of $K_1$, $K_3$, and all stars $K_{1,n-1}$. Let $G$ and $H$ be two graphs, each
with at least 4 vertices. Suppose that $G$ contains $P_4$ and $H$ does not.
Note that $H=K_{1,n-1}$.
If $G$ contains a subgraph $K_3$, Spoilers pebbles it and wins because
there is no $K_3$ in $H$. Assume that $G$ has no $K_3$.
Let $a_1a_2a_3a_4$ be a path in $G$. In the first round Spoiler
pebbles $a_2$.
If Duplicator responds with the central vertex of the star $H$,
Spoiler wins by pebbling $a_4$, which is not adjacent to $a_2$
due to the absence of $K_3$ in $G$.
If Duplicator responds with a leaf of $H$,
Spoiler wins by pebbling $a_1$ and $a_3$.

It remains to prove the lower bound. Consider
$G=J_{3,n}$ and $H=J_{2,n}=K_{1,n+1}$.
These graphs have $n+3$ and $n+2$ vertices respectively, and the parameter $n$
can be chosen arbitrarily large.
Note that $G$ contains $P_4$ as a subgraph, while $H$ does not.
As easily seen, $D(G,H)\ge3$.
\end{proof}

\begin{lemma}
$W_v(P_4)=2$.
\end{lemma}

\begin{proof}
Suppose that $G$ contains a (not necessarily induced) subgraph $P_4$ and $H$ does not.
Let $H$ have more than 3 vertices; then $H=K_{1,n}$.
We have to show that Spoiler wins the 2-pebble game on $G$ and $H$
making a bounded number of moves (that does not depend on how large $G$ and $H$ are).
In the first round he pebbles the central vertex in $H$.
Suppose that $G$ has a universal vertex for else Spoiler wins in the next round.
Not to lose in the next round, Duplicator pebbles a universal vertex $u$ in $G$.
If $G$ has yet another universal vertex, Spoiler pebbles it and wins in the 3rd round
by reusing the first pebble. Assume, therefore, that $u$ is the only universal
vertex in $G$. Since $G$ is not a star graph, it contains a vertex $v\ne u$
having a neighbor $w\ne u$. Spoiler pebbles $v$ in the second round.
Duplicator responds with leaf in $H$. In the third round Spoiler moves the pebble
from $u$ to $w$. Duplicator is forced to respond with pebbling the central vertex in $H$.
Spoiler wins in the fourth round by moving the pebble from $v$ to a vertex non-adjacent
with~$w$.
\end{proof}

\paragraph{The cycle subgraph ($C_4$).}

To prove part 3 of Theorem \ref{thm:four}, we need some properties of random regular graphs,
established by Bollob\'as \cite{Bollobas-b} and  Wormald \cite{Wor99}.
We collect them in the following lemma.

\begin{lemma}\label{lem:Warmald}
Let $d,g\ge3$ be fixed, and $dn$ be even.
Let $\mathcal{G}_{n,d}$ denote a random $d$-regular graph on $n$ vertices.
\begin{enumerate}
\item
$\mathcal{G}_{n,d}$ is $d$-connected with high probability (see \cite[Section 7.6]{Bollobas-b} or \cite[Section~2.6]{Wor99}).
\item
$\mathcal{G}_{n,d}$ has girth $g$ with probability that converges to a limit different from 0 and from 1
(see \cite[Corollary 2.19]{Bollobas-b} or \cite[Theorem~2.5]{Wor99}).
\item
$\mathcal{G}_{n,d}$ has no non-trivial automorphism with high probability (see \cite[Theorem~9.10]{Bollobas-b}).
\end{enumerate}
\end{lemma}

\begin{lemma}
$D_\kappa(C_4)=4$.
\end{lemma}

\begin{proof}
Fix $k$ as large as desired. Lemma \ref{lem:Warmald} provides us with
$k$-regular graphs $G$ and $H$ such that
\begin{enumerate}
\item
$G$ has girth exactly 4, and $H$ has girth strictly more than~4;
\item
both $G$ and $H$ are $k$-connected;
\item
$G$ has no non-trivial automorphism;
\item
both $G$ and $H$ have no less than $k^2+2$ vertices.
\end{enumerate}
It suffices to show that $D(G,H)>3$. To this end, we describe
a strategy allowing Duplicator to win the 3-round game.

As usually, we denote the vertices pebbled in $G$ and $H$ in the $i$-th round by $x_i$ and $y_i$
respectively. We desribe Duplicator's move assuming that Spoiler has moved in $G$;
the other case is symmetric with the only exception that will occur in the end of our analysis.

\smallskip

\textit{1st round}.
Duplicator pebbles an arbitrary vertex $y_1$ in~$H$.

\smallskip

\textit{2nd round}.
Duplicator pebbles $y_2$ such that $d(y_1,y_2)=d(x_1,x_2)$
if $d(x_1,x_2)\le2$ and $d(y_1,y_2)\ge3$ if $d(x_1,x_2)\ge3$.
If $d(x_1,x_2)=2$, such a vertex exists by the assumptions on the girth.
If $d(x_1,x_2)\ge3$, it exists because $H$ has more than $1+k+k(k-1)$ vertices.

\smallskip

\textit{3rd round}.
If $x_3$ is adjacent neither to $x_1$ nor to $x_2$, then Duplicator
pebbles $y_3$ adjacent neither to $y_1$ nor to $y_2$, which exists
because there are more than $2k+2$ vertices in the graph.
Assume, therefore, that $x_3$ is adjacent to at least one of $x_1$ and~$x_2$.

\Case1{$d(x_1,x_2)=1$ or $d(x_1,x_2)\ge3$.}
Duplicator wins by pebbling $y_3$ adjacent either to $y_1$ or to $y_2$ depending on
the neighborhood of $x_3$. Note that $x_3$ cannot be adjacent to both $x_1$ and $x_2$.
If $d(x_1,x_2)=1$, this follows from the assumption on the girth.

\Case2{$d(x_1,x_2)=2$.}
If $x_3$ is adjacent to both $x_1$ and $x_2$, Duplicator wins by pebbling
$y_3$ adjacent to both $y_1$ and $y_2$.
It remains to consider the subcase when $x_3$ is adjacent to exactly one of $x_1$ and $x_2$, say,
to $x_2$. Then Duplicator pebbles an arbitrary vertex $y_3$ adjacent to $y_2$.
Note that $y_3$ and $y_1$ are not adjacent because $H$ has girth larger than 4.
However, this argument does not work when Spoiler plays the 3rd round in $H$.

Thus, the following case takes more care: $d(x_1,x_2)=d(y_1,y_2)=2$ and Spoiler pebbles $y_3$
adjacent to $y_2$ and not adjacent to $y_1$. Fortunately, Duplicator anyway has
a choice of $x_3$ adjacent to $x_2$ and not adjacent to $x_1$. Indeed, if all
neighbots of $x_2$ were adjacent also to $x_1$, then the transposition of $x_2$ and $x_1$
would be an automorphism of $G$, contradicting the assumption.
\end{proof}

\begin{lemma}
$W_\kappa(C_4)\ge3$.
\end{lemma}

\begin{proof}
Fix an arbitrarily large $d$. Consider a random $d$-regular graph $\mathcal{G}_{n,d}$
with sufficiently large number of vertices $n$. By Lemma \ref{lem:Warmald},
this graph is $d$-connected with high probability. Also, $\mathcal{G}_{n,d}$ has girth 4 with nonzero probability,
and as well it has girth 5 with nonzero probability. This yields us
two $d$-connected graphs $G$ and $H$ of girth 4 and 5 respectively.
Note that both graphs have neither universal nor isolated vertices. 
This readily implies that $W(G,H)>2$.
\end{proof}

The proof of Theorem \ref{thm:four} is complete.

\section{Concluding remarks and questions}\label{s:que}

\subsection{Definitions over 2-connected graphs}

Suppose that a pattern graph $F$ is 2-connected. Note that $G$ contains an $F$ subgraph if and only if
a 2-connected component of $G$ contains such a subgraph.
This motivates a relaxation of our setting to considering definitions of
$\subgr F$ over \emph{2-connected} (rather than \emph{connected}) graphs.
In particular, let $W'_\tw(F)$ denote the minimum $m$ for which there are
a formula $\Phi$ of variable width $m$ and a number $k$
such that $G\models\Phi$ exactly when $G\in\subgr F$, for all
2-connected $G$ of treewidth at least $k$.
Note that \subgr F is recognizable in time $O(n^{W'_\tw(F)})$ and
\begin{equation}
  \label{eq:W'}
W_\kappa(F)\le W'_\tw(F)\le W_\tw(F)=v(F).
\end{equation}

Consider, for example, $F=C_4$. Since $C_4$ is 2-connected, we know that $W_\tw(C_4)=4$,
but do not know whether or not $W_\kappa(C_4)=4$.
With some extra effort, we can determine the value of~$W'_\tw(C_4)$.

\begin{proposition}
$W'_\tw(C_4)=4$.
\end{proposition}

\begin{proof}
  We first exhibit a graph $G$ containing a $C_4$
and a $C_4$-free graph $H$ such that $W(G,H)>3$. Let $G$ be the cube graph,
i.e., $G=(K_2)^3$ (the Cartesian power), and $H=C_6$.
Duplicator wins the $3$-pebble game on these graphs whatever the number of rounds, that is,
\begin{equation}
  \label{eq:cube-C6}
W((K_2)^3,C_6)>3.
\end{equation}
In order to see this, note that both graphs have diameter 3. A winning strategy for Duplicator
is based on the following observation. Let $x_1,x_2\in V((K_2)^3)$ and $y_1,y_2\in V(C_6)$.
Suppose that $d(x_1,x_2)=d(y_1,y_2)$, where $d(u,v)$ denotes the distance between two
vertices. Then for every $x\in V((K_2)^3)$ there is a $y\in V(C_6)$, and vice versa,
such that $d(x,x_1)=d(y,y_1)$ and $d(x,x_2)=d(y,y_2)$. This allows Duplicator
to keep the distances between the same pair of pebbles in $(K_2)^3$ and $C_6$ always equal.\footnote{%
In \cite{induced}, we show that that a similar strategy can be used on any pair of
\emph{distance-regular} graphs satisfying certain similarity conditions.}

We now show that there are many other pairs of 2-connected graphs
$G\in\subgr{C_4}$ and $G\notin\subgr{C_4}$ with $W(G,H)>3$ and,
moreover, they can have arbitrarily large treewidth.
At this point, it is useful to note that both $(K_2)^3$ and $C_6$
are antipodal in the sense \cite{BermanK88}.
A connected graph is defined in \cite{BermanK88} to be \emph{antipodal}
if for every vertex $v$ there is a unique vertex $\bar v$ of maximum distance from $v$.

Let $A$ be an antipodal graph with designated pair of antipodal vertices $a$ and $\bar a$.
Let $B$ be an arbitrary graph. We construct the product graph $\dobut BA$ as follows:
\begin{itemize}
\item
Subdivide each edge of $B$ by two new vertices $b$ and $b'$.
\item
Replace each edge $bb'$ with a copy of $A$ identifying $b$ with $a$ and $b'$ with~$\bar a$.
\end{itemize}
Though the second step does not seem to be symmetric with respect to $b$ and $b'$,
this construction is actually isomorphism-invariant under interchanging $b$ and $b'$.
This follows from the known fact (see \cite{BermanK88}) that the map taking each vertex $v$
to its antipose $\bar v$ is an automorphism of an antipodal graph.
In fact, we will apply this construction only to factor graphs $A=(K_2)^3$ and $A=C_6$,
where we have isomorphism invariance not only with respect to swapping $b$ and $b'$
but even with respect to the choice of the antipodal pair $a,\bar a$;
the latter because both $(K_2)^3$ and $C_6$ are distance transitive.

Fix $k$ as large as desired.
Take a 2-connected graph $B$ of treewidth at least $k$ and consider
graphs $G=\dobut B{(K_2)^3}$ and $G=\dobut B{C_6}$.
These graphs have the following properties:
\begin{itemize}
\item
Both $G$ and $H$ are 2-connected;
\item
$\tw(G)\ge\tw(B)\ge k$ and $\tw(H)\ge\tw(B)\ge k$
(because $B$ is a minor of both $G$ and $H$, see  \cite{Diestel});
\item
$G$ has girth 4, and $H$ has girth 6.
\end{itemize}
It remains to notice that $W(G,H)>3$. Duplicator wins the 3-pebble game
using the following strategy:
\begin{itemize}
\item
Whenever Spoiler pebbles a vertex of the subdivided $B$ in $G$ or $H$,
Duplicator responds with the same vertex of the subdivided $B$ in the other graph;
\item
Whenever Spoiler moves in a copy of $(K_2)^3$ in $G$, Duplicator responds
by playing in the corresponding copy of $C_6$ in $H$ (and vice versa)
using her winning strategy in the game on $(K_2)^3$ and $C_6$ that
was described above in the the proof of the inequality~\refeq{cube-C6}.
\end{itemize}
Note that the two rules above are consistent because Duplicator's
strategy in the the proof of the inequality~\refeq{cube-C6} respects
the antipodality relation.
\end{proof}

It would be instructive to see an example of a 2-connected pattern graph $F$
such that $W'_\tw(F)<v(F)$.

\subsection{Further questions}
\mbox{}

\que
The first inequality in \refeq{W'} motivates an interest in lower bounds on $W_\kappa(F)$
for 2-connected pattern graphs $F$. For instance, one can show that $D_\kappa(C_\ell)\ge\log_3\ell$.
On the other hand, we currently cannot disprove even that $W_\kappa(C_\ell)=O(1)$.
For small graphs, we do not know whether or not $W_\kappa(C_4)=4$
and $W_\kappa(K_4\setminus e)=4$; cf.\ Table~\ref{fig:small}.

\que
Is the bound $D_v(F)>v(F)/2$, given by Theorem \ref{thm:ell/2}, tight?
On the other hand, currently we cannot disprove even that $D_v(F)\ge v(F)-O(1)$.

\que
It is known \cite{ChandranS05} that $\tw(F)\ge e(F)/v(F)$.
Can one improve Theorem \ref{thm:e/v} to $D_\kappa(F)\ge\tw(F)$?

\que
The parameters $D_\logic(\classc)$ and $W_\logic(\classc)$
have been studied in various contexts also for other graph properties $\classc$
and other logics $\logic$.
We refer an interested reader to~\cite{Dawar98,Turan84}.
In many cases, it would be interesting to compare $D_\logic(\classc)$ and $D_{\logic'}(\classc)$
(or $W_\logic(\classc)$ and $W_{\logic'}(\classc)$) for various different logics $\logic$
and $\logic'$ and the same property~$\classc$.

\subparagraph*{Acknowledgements.}

We would like to thank Tobias Müller for his kind hospitality during
the Workshop on Logic and Random Graphs in the Lorentz Center
(August 31 -- September 4, 2015), where this work was originated.


\begin{thebibliography}{10}

\bibitem{AlonS16}
N.~Alon and J.~H. Spencer.
\newblock {\em The probabilistic method}.
\newblock John Wiley \& Sons, 2016.

\bibitem{AlonYZ95}
N.~Alon, R.~Yuster, and U.~Zwick.
\newblock Color-coding.
\newblock {\em J. {ACM}}, 42(4):844--856, 1995.

\bibitem{Amano10}
K.~Amano.
\newblock {$k$-Subgraph} isomorphism on {\ac0} circuits.
\newblock {\em Computational Complexity}, 19(2):183--210, 2010.

\bibitem{BermanK88}
A.~Berman and A.~Kotzig.
\newblock Cross-cloning and antipodal graphs.
\newblock {\em Discrete Math.}, 69(2):107--114, 1988.

\bibitem{Bodlaender96}
H.~L. Bodlaender.
\newblock A linear-time algorithm for finding tree-decompositions of small
  treewidth.
\newblock {\em {SIAM} J. Comput.}, 25(6):1305--1317, 1996.

\bibitem{Bollobas-b}
B.~Bollob\'as.
\newblock {\em Random graphs}, volume~73 of {\em Cambridge Studies in Advanced
  Mathematics}.
\newblock Cambridge University Press, Cambridge, second edition, 2001.

\bibitem{ChandranS05}
L.~S. Chandran and C.~R. Subramanian.
\newblock Girth and treewidth.
\newblock {\em J. Comb. Theory, Ser. {B}}, 93(1):23--32, 2005.

\bibitem{ChekuriC14}
C.~Chekuri and J.~Chuzhoy.
\newblock Polynomial bounds for the grid-minor theorem.
\newblock In {\em Proc. of the 46th ACM Symposium on Theory of Computing
  (STOC'14)}, pages 60--69, 2014.

\bibitem{ChenHKX06}
J.~Chen, X.~Huang, I.~A. Kanj, and G.~Xia.
\newblock Strong computational lower bounds via parameterized complexity.
\newblock {\em J. Comput. Syst. Sci.}, 72(8):1346--1367, 2006.

\bibitem{Courcelle90}
B.~Courcelle.
\newblock The monadic second-order logic of graphs {I}. {R}ecognizable sets of
  finite graphs.
\newblock {\em Inf. Comput.}, 85(1):12--75, 1990.

\bibitem{Dawar98}
A.~Dawar.
\newblock A restricted second order logic for finite structures.
\newblock {\em Inf. Comput.}, 143(2):154--174, 1998.

\bibitem{Diestel}
R.~Diestel.
\newblock {\em Graph theory}.
\newblock New York, NY: Springer, 2000.

\bibitem{FloderusKLL15}
P.~Floderus, M.~Kowaluk, A.~Lingas, and E.~Lundell.
\newblock Induced {S}ubgraph {I}somorphism: Are some patterns substantially
  easier than others?
\newblock {\em Theoretical Computer Science}, 605:119--128, 2015.

\bibitem{Gall14}
F.~L. Gall.
\newblock Powers of tensors and fast matrix multiplication.
\newblock In {\em Proc. of the Int. Symposium on Symbolic and Algebraic
  Computation (ISSAC'14)}, pages 296--303. {ACM}, 2014.

\bibitem{GraedelG15}
E.~Gr{\"{a}}del and M.~Grohe.
\newblock Is polynomial time choiceless?
\newblock In {\em Fields of Logic and Computation {II}, Essays Dedicated to
  Yuri Gurevich on the Occasion of His 75th Birthday}, volume 9300 of {\em
  Lecture Notes in Computer Science}, pages 193--209. Springer, 2015.

\bibitem{Immerman-book}
N.~Immerman.
\newblock {\em Descriptive complexity}.
\newblock New York, NY: Springer, 1999.

\bibitem{JLR-book}
S.~{Janson}, T.~{{\L}uczak}, and A.~{Ruci\'nski}.
\newblock {\em Random graphs}.
\newblock New York, Berlin: Wiley, 2000.

\bibitem{KawarabayashiR16}
K.~Kawarabayashi and B.~Rossman.
\newblock An excluded-minor approximation of tree-depth.
\newblock Manuscript, 2016.

\bibitem{KouckyLPT06}
M.~Kouck{\'{y}}, C.~Lautemann, S.~Poloczek, and D.~Th{\'{e}}rien.
\newblock Circuit lower bounds via {Ehrenfeucht-Fra{\"\i}ss{\'e}} games.
\newblock In {\em Proc. of the 21st Ann. IEEE Conf. on Computational Complexity
  (CCC'06)}, pages 190--201, 2006.

\bibitem{LiRR14}
Y.~Li, A.~A. Razborov, and B.~Rossman.
\newblock On the {AC$^0$} complexity of {S}ubgraph {I}somorphism.
\newblock {\em SIAM J. Comput.}, 46(3):936--971, 2017.

\bibitem{Libkin04}
L.~Libkin.
\newblock {\em Elements of Finite Model Theory}.
\newblock Texts in Theoretical Computer Science. An {EATCS} Series. Springer,
  2004.

\bibitem{NesetrilP85}
J.~Ne\v{s}et\v{r}il and S.~Poljak.
\newblock On the complexity of the subgraph problem.
\newblock {\em Commentat. Math. Univ. Carol.}, 26:415--419, 1985.

\bibitem{Olariu88}
S.~Olariu.
\newblock Paw-free graphs.
\newblock {\em Inf. Process. Lett.}, 28:53--54, 1988.

\bibitem{PikhurkoV11}
O.~Pikhurko and O.~Verbitsky.
\newblock Logical complexity of graphs: a survey.
\newblock In M.~Grohe and J.~Makowsky, editors, {\em Model theoretic methods in
  finite combinatorics}, volume 558 of {\em Contemporary Mathematics}, pages
  129--179. American Mathematical Society (AMS), Providence, RI, 2011.

\bibitem{Pinsker73}
M.~S. Pinsker.
\newblock On the complexity of a concentrator.
\newblock In {\em Annual 7th International Teletraffic Conference}, pages
  318/1--318/4, 1973.

\bibitem{Rossman08}
B.~Rossman.
\newblock On the constant-depth complexity of $k$-clique.
\newblock In {\em Proc. of the 40th Ann. {ACM} Symposium on Theory of Computing
  (STOC'08)}, pages 721--730. {ACM}, 2008.

\bibitem{Rossman16}
B.~Rossman.
\newblock An improved homomorphism preservation theorem from lower bounds in
  circuit complexity.
\newblock {\em {SIGLOG} News}, 3(4):33--46, 2016.

\bibitem{Rossman-talk}
B.~Rossman.
\newblock Lower bounds for {S}ubgraph {I}somorphism and consequences in
  first-order logic.
\newblock Talk in the Workshop on \emph{Symmetry, Logic, Computation} at the
  Simons Institute, Berkeley, November 2016.
\newblock \url{https://simons.berkeley.edu/talks/benjamin-rossman-11-08-2016}.

\bibitem{Schweikardt13}
N.~Schweikardt.
\newblock A short tutorial on order-invariant first-order logic.
\newblock In {\em Proc. of the 8th Int. Computer Science Symposium in Russia
  (CSR'13)}, volume 7913 of {\em Lecture Notes in Computer Science}, pages
  112--126. Springer, 2013.

\bibitem{Turan84}
G.~Tur{\'{a}}n.
\newblock On the definability of properties of finite graphs.
\newblock {\em Discrete Mathematics}, 49(3):291--302, 1984.

\bibitem{VZh16}
O.~Verbitsky and M.~Zhukovskii.
\newblock The descriptive complexity of {S}ubgraph {I}somorphism without
  numerics.
\newblock In P.~Weil, editor, {\em Proc. of the 12th Int. Computer Science
  Symposium in Russia (CSR'17)}, volume 10304 of {\em Lecture Notes in Computer
  Science}, pages 308--322. Springer, 2017.

\bibitem{induced}
O.~Verbitsky and M.~Zhukovskii.
\newblock On the first-order complexity of {I}nduced {S}ubgraph {I}somorphism.
\newblock In {\em Proc.\ of the 26th EACSL Annual Conference on Computer
  Science Logic (CSL 2017)}, volume~82 of {\em LIPIcs–Leibniz International
  Proceedings in Informatics}, pages 40:1--40:16, 2017.

\bibitem{Wor99}
N.~Wormald.
\newblock Models of random regular graphs.
\newblock In {\em Surveys in Combinatorics}, pages 239--298. Cambridge
  University Press, 1999.

\bibitem{Zhuk}
M.~Zhukovskii.
\newblock Zero-one $k$-law.
\newblock {\em {Discrete Mathematics}}, 312:1670--1688, 2012.

\end{thebibliography}
\end{document}